 \documentclass[draftclsnofoot,onecolumn,12pt]{IEEEtran}

\ifCLASSINFOpdf
\else
\fi
\usepackage{graphics}
\usepackage{graphicx}
\usepackage{epsfig}
\usepackage{subfigure}
\usepackage{amssymb}
\usepackage{amsmath}
\usepackage{amsfonts}
\usepackage{stmaryrd}
\usepackage{txfonts}
\usepackage{mathrsfs}
\usepackage{latexsym, bm}
\usepackage{overpic}
\newtheorem{theorem}{Theorem}

\newtheorem{definition}{Definition}
\newtheorem{lemma}{Lemma}
\newtheorem{corollary}{Corollary}

\newtheorem{remark}{Remark}
\newtheorem{example}{Example}

\newtheorem{proof}{Proof}

\DeclareMathOperator{\diag}{diag}

\begin{document}
%
% paper title
% can use linebreaks \\ within to get better formatting as desired
\title{Non-Fragility and Partial Controllability of Multi-Agent Systems}
%
%
% author names and IEEE memberships
% note positions of commas and nonbreaking spaces ( ~ ) LaTeX will not break
% a structure at a ~ so this keeps an author's name from being broken across
% two lines.
% use \thanks{} to gain access to the first footnote area
% a separate \thanks must be used for each paragraph as LaTeX2e's \thanks
% was not built to handle multiple paragraphs
%

\author{Bin~Zhao,~%\IEEEmembership{Member,~IEEE,}
       Yongqiang~Guan,~%\IEEEmembership{Member,~IEEE}
       and~Long~Wang~%\IEEEmembership{Member,~IEEE}
        %and~Jane~Doe,~\IEEEmembership{Life~Fellow,~IEEE}% <-this % stops a space
\thanks{This work was supported by NSFC (61375120).
%This work is supported by National Natural Science Foundation of China (NNSF) under Grant (No. 61203374, 61375120, 61374062), Beijing Natural Science Foundation (No. 4142031) and Science Foundation of Shandong Province for Distinguished Young Scholars (No. JQ201419).
}
\thanks{B. Zhao and L. Wang are with the Center for Systems and Control, College of Engineering, Peking University, Beijing 100871, China
 (e-mail: bigbin@pku.edu.cn; longwang@pku.edu.cn)}
\thanks{Y. Guan is with the Center for Complex Systems, School of Mechano-Electronic Engineering, Xidian University, Shaanxi, Xi¡¯an, 710071, China
 (e-mail: guan-jq@163.com)
 %e-mail: (see http://www.michaelshell.org/contact.html).
 }% <-this % stops a space
%\thanks{Z. Ji is with College of Automation Engineering, Qingdao University, Qingdao, 266071, China
% (e-mail:jizhijian@pku.edu.cn)}

% <-this % stops a space
%\thanks{Manuscript received April 19, 2005; revised January 11, 2007.}
}

\maketitle

\begin{abstract}
%\boldmath
Controllability of multi-agent systems is determined by the interconnection topologies. In practice, losing agents can change the topologies of multi-agent systems, which may affect the controllability. This paper studies non-fragility of controllability influenced by losing agents. In virtue of the concept of cutsets, necessary and sufficient conditions are established from a graphic perspective, for strong non-fragility and weak non-fragility of controllability, respectively. For multi-agent systems which contain important agents, partial controllability is proposed in terms of the concept of controllable node groups, and necessary and sufficient criteria are established for partial controllability. Moreover, partial controllability preserving problem is proposed. Utilizing the concept of compressed graphs, this problem is transformed into finding the the minimal $\mathbf{\langle s,t\rangle}$ vertex cutsets of the interconnection graph, which has a polynomial-time complexity algorithm for the solution. Several constructive examples illuminate the theoretical results.
\end{abstract}


% For peer review papers, you can put extra information on the cover
% page as needed:
% \ifCLASSOPTIONpeerreview
% \begin{center} \bfseries EDICS Category: 3-BBND \end{center}
% \fi
%
% For peerreview papers, this IEEEtran command inserts a page break and
% creates the second title. It will be ignored for other modes.
\IEEEpeerreviewmaketitle

\section{Introduction}
In recent years, distributed coordination control of multi-agent systems (MASs) has become an important topic due to the wide connections between MASs and numerous subjects. Researches in this area include consensus problem \cite{Xiao10,Wang10}, formation control \cite{Xiao09}, flocking \cite{Saber06}, controllability and stabilizability \cite{Guan13,Guan16}, etc.

Controllability of MASs was proposed by Tanner for the first time in \cite{Tanner04}, where a necessary and sufficient condition was presented through the Laplacian matrix and the corresponding eigenvalues. Wang et al. studied controllability of MASs with high-order dynamics and generic linear dynamics, and showed that controllability is congruously determined by the interconnection topology \cite{Wang09}. Afterwards, researchers attempted to investigate controllability of MASs from algebraic point of view and graphic perspective. For example, Zhao et al. designed a leader selection algorithm using the algebraic properties of the Laplacian matrix \cite{Zhao16}; Ji et al. proposed a construction procedure for uncontrollable topologies \cite{Ji15}; and Rahmani et al. provided some necessary conditions for controllability utilizing the equitable partition of the interconnection topology \cite{Rahmani09}. In addition, interesting methods were developed for controllability of some special graphs, e.g., tree graphs via analysing the leaders' role with downer branches \cite{Ji12}; paths, cycles and grid graphs via simple rules from number theory \cite{Parlangeli12,Notarstefano13}, etc. A parallel research line in this field is structural controllability of MASs, which was investigated under various models \cite{Zamani009,Partovi10}, etc. The relationship between controllability and structural controllability was studied in \cite{Goldin13}, and controllability improvement for structurally controllable systems was discussed in \cite{Zhao16}.

However, the previous results only consider necessary and/or sufficient conditions for controllability. In practice, for MASs, losing agents is a common phenomenon. For example, in robot systems and vehicle systems, malfunction of some units may appear during the formation process; in biological systems, individuals of a species might be dead during the migration; and in social groups, members of an organization may quit at any time. On the one hand, losing agents and failure of communication links may influence controllability and structural controllability of MASs. Therefore, in \cite{Jafari11}, an optimal selection of the fewest leaders was shown to improve the reliability of MASs in terms of controllability. In \cite{Rahimian13}, the robustness of structural controllability was investigated against the failure of agents and communication links simultaneously. On the other hand, many MASs contain important agents, whose dynamic behavior needs to be controlled properly (e.g., in ant groups, only queens and males have the ability to breed offspring, which are the core parts of the colony, while the ergates are of less importance \cite{Brain}). For an MAS consisting of core agents and less important ones, if we intend to control the important part, it is not essential that all the agents in the system being controllable simultaneously, which derives the concept of partial controllability. Considering of the influence of losing agents, how to preserve partial controllability of MASs becomes a meaningful issue.

Motivated by the above analysis, this paper studies non-fragility of controllability and partial controllability of MASs. The main contributions of this paper are threefold: \\
i) The concept of non-fragility is proposed for controllability. Utilizing the notion of cutsets, necessary and sufficient graphic conditions of strong non-fragility and weak non-fragility of MASs are established, respectively. The difference between structural controllability and non-fragility of controllability is clarified.\\
ii) The concept of controllable node groups is proposed for partial controllability of MASs. Basic criteria of partial controllability are provided. Especially, the correspondence between the controllable node groups in the system and the linearly independent rows of the controllability matrix is clearly revealed, which lays the foundation for the research of partial controllability. \\
iii) The problem of preserving partial controllability against losing agents is proposed. Utilizing the concept of compressed graphs, the problem is equivalently converted into finding the minimal $\langle s,t\rangle$ vertex cutsets of the interconnection graph, which has been proved to have a polynomial-time complexity algorithm for the solution.
%%
%%This paper is organized as follow: Section 2 introduces the basic concepts and mathematical tools for this paper. Main results on partial controllability, as partial controllability criteria, states of the uncontrollable nodes, and the algorithm of leader selection problem, are shown in three subsections of Section 3, respectively. Numerical simulation examples are provided in Section 4 to illustrate theoretical results. Conclusions are drawn in Section 5.

$\mathbf{Notations:}$ $||\cdot||$ represents the Euclidean $2$-norm of a vector. $\diag(a_1,a_2,\cdots,a_n)$ is the diagonal matrix with principal diagonals $a_1,a_2,\cdots,a_n$. The set of $n$-dimensional real vectors is denoted by $\mathbb{R}^n$. $|S|$ represents the cardinality of set $S$. $C_n^p$ denotes the combination number, selecting $p$ items from $n$ item. $S/T$ means the set of all the elements in $S$ but not in $T$.

%%%%%%%%%%%%%%%%%%%%%%%%%%%%%%%%%%%%%%%%%%%%%%%%%%%%%%%%%%%%%%%%%%%%%%%%%%%%%%
\section{Preliminaries}
\subsection{Graph theory}

An undirected graph $\mathbb{G}=(\mathbb{V},\mathbb{E})$ consists of a vertex set $\mathbb{V}=\{v_1,v_2,\cdots,v_n\}$, and an edge set $\mathbb{E}\subseteq \mathbb{V}\times\mathbb{V}$. In graph $\mathbb{G}$, $e_{ij}\in\mathbb{E}$ if and only if $e_{ji}\in\mathbb{E}$, and $v_i,~v_j$ are said to be adjacent with each other. The neighbor set of $v_j$ is denoted by $N_j=\{v_i\in\mathbb{V}|(v_i,v_j)\in\mathbb{E}\}$. The adjacency matrix of $\mathbb{G}$ is $A(\mathbb{G})=[a_{ij}]\in \mathbb{R}^{n\times n}$, where $a_{ij}> 0$ is the weight of edge $e_{ji}$ (as well as $e_{ij}$), and $a_{ij}=0$ if $(v_j,v_i)\notin\mathbb{E}$. The Laplacian matrix of $\mathbb{G}$ is $L(\mathbb{G})=D-A$, $D=\diag(d_1,d_2,\cdots,d_n)$ where $d_k=\sum\limits_{i=1,i\neq k}^n a_{ki},~k=1,2,\cdots,n$. Graph $\mathbb{G}'=(\mathbb{V}',\mathbb{E}')$ is called a subgraph of $\mathbb{G}$ if $\mathbb{V}'\subseteq\mathbb{V}, \mathbb{E}'\subseteq\mathbb{E}$. Removing a vertex $v$ from $\mathbb{G}$ means deleting $v$ and all the edges connected with $v$ in $\mathbb{G}$, and the remaining subgraph is denoted as $\mathbb{G}-v$. Removing a vertex set $\mathbb{V}'$ from $\mathbb{G}$ means deleting all the vertexes in $\mathbb{V}'$ and all the edges connected with any vertex in $\mathbb{V}'$, and denote the remaining subgraph as $\mathbb{G}-\mathbb{V}'$. $\mathbb{G}'=(\mathbb{V}',\mathbb{E}')$ is said to be the induced subgraph of $\mathbb{G}$ by $\mathbb{V}'$, if $\mathbb{G}'=\mathbb{G}-\mathbb{V}/\mathbb{V}'$. A path between $v_i$ and $v_j$ is a subgraph of $\mathbb{G}$, whose vertex set is $\{v_i,v_{k_1},\cdots,v_{k_r},v_j\}$ and the edge set is $\{(v_i,v_{k_1}),(v_{k_1},v_{k_2}),\cdots,(v_{k_{r-1}},v_{k_r}),(v_{k_r},v_j)\}$, where $0\leq r,~1\leq i,j,k_1,\cdots,k_r\leq n$, and no two nodes in $\{v_i,v_{k_1},\cdots,v_{k_r},v_j\}$ are same. For two vertexes $v_i\neq v_j$, we say they are in the same connected component if there exists a path between them, otherwise, they are in different connected components. The number of connected components of $\mathbb{G}$ is denoted as $p(\mathbb{G})$. $\mathbb{G}$ is said to be connected if $p(\mathbb{G})=1$.
% and the corresponding weights of the edges in $\mathbb{G}'$ equals those in $\mathbb{G}$
\begin{definition}
\cite{Zhang13} The distance partition of a connected graph $\mathbb{G}$ relative to node $v$ consists of a series of node sets $D_0,D_1,D_2,\cdots,D_l$, where $D_0=\{v\}$, $D_i=\{w\in \mathbb{V}|~w$ is adjacent with some node in $D_{i-1}$, but not adjacent with any node in $D_0,\cdots,D_{i-2}\}$. $\bigcup\limits_i D_i=\mathbb{V},i=0,1,2,\cdots,l$. Especially, for a connected graph, the minimal $s$ such that $\bigcup\limits_{i=1}^s D_i=\mathbb{V}$ is said to be the length of the graph (relative to $v$).
\end{definition}

\begin{figure}
\begin{minipage}[t]{0.65\linewidth}
\centering
\includegraphics[width=2.7in]{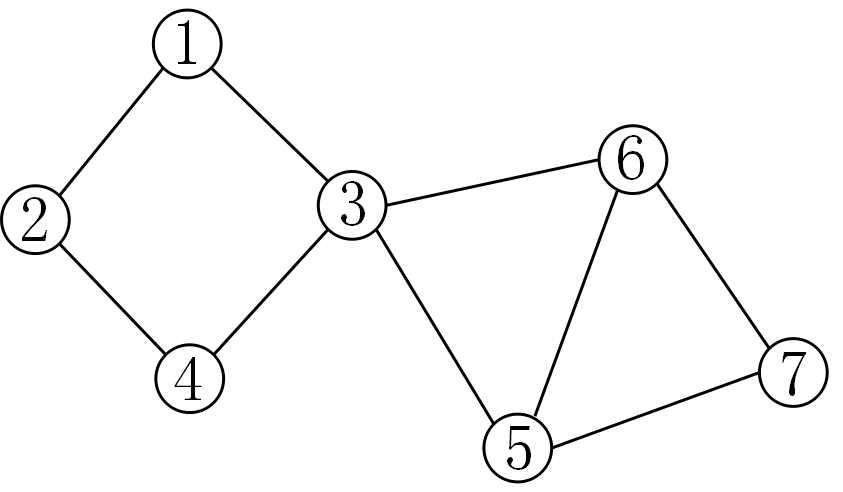}
\caption{A graph consists of 7 nodes.}\label{F1}
\end{minipage}%
~
\begin{minipage}[t]{0.35\linewidth}
\centering
\includegraphics[width=1.5in]{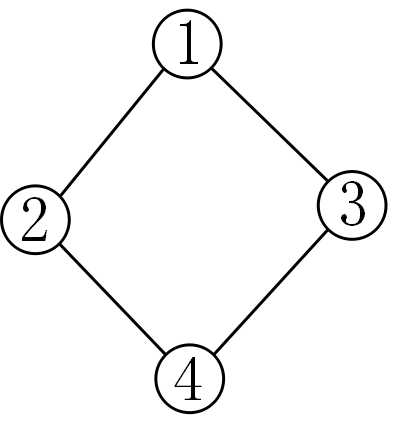}
\caption{The compressed graph of Fig. \ref{F1}.}\label{F10}
\end{minipage}
\end{figure}
%\begin{example}
%For Fig. \ref{F1}, the distance partition relative to $v_1$ is: $D_0=\{v_1\}, D_1=\{v_2,v_3\}, D_2=\{v_4,v_5,v_6\}, D_3=\{v_7\}$. The length of the graph is $3$.
%\end{example}
%
%
%Next we introduce the most important graph theory concept in this paper, called ``cutset''.
\begin{definition}\label{CS}
\cite{Wilson} For a graph $\mathbb{G}=(\mathbb{V},\mathbb{E})$, vertex set $\mathbb{V}'\subset \mathbb{V}$ is said to be a cutset of $\mathbb{G}$, if when all the nodes in $\mathbb{V}'$ are removed, the subgraph $\mathbb{G}-\mathbb{V}'$ contains more connected components than $\mathbb{G}$, i.e., $p(\mathbb{G}-\mathbb{V}')>p(\mathbb{G})$, whereas removing any proper subset $\mathbb{V}''\subset \mathbb{V}'$ will not increase the connected components of $\mathbb{G}$. The minimal cutset is a cutset of $\mathbb{G}$ that contains the fewest vertexes. If a cutset contains only one vertex $v$, we call $v$ a cut vertex.
\end{definition}

Generally speaking, a cutset is a set of nodes, when removed, will lead to more connected components than in the original graph. For example, in Fig. \ref{F1}, removing a single vertex $v_3$ will break the connectivity, therefore $v_3$ is a cut vertex of Fig. \ref{F1}. Apparently, neither of $v_5,v_6$ is a cut vertex, however, removing vertex set $\{v_5,v_6\}$ makes $v_7$ separated from the other nodes, which means $v_5$ and $v_6$ form a cutset of Fig. \ref{F1}.

\begin{definition}\label{stCS}
\cite{Abel} An $\langle s,t\rangle$ vertex cutset is a set of vertexes excluding $s$ and $t$, that all paths connecting $s$ and $t$ pass through at least one vertex in this set. The set is said to be minimal if it contains the fewest vertexes among all the possible $\langle s,t\rangle$ vertex cutsets.
\end{definition}

\begin{remark}
In Fig. \ref{F1}, sets $\{v_2,v_3\}$ and $\{v_2,v_3,v_7\}$ are both $\langle v_1,v_4\rangle$ vertex cutsets of the graph, while $\{v_2,v_3\}$ is a minimal one. Obviously, if a node set $\tilde{\mathbb{V}}$ forms an $\langle s,t\rangle$ vertex cutset of graph $\mathbb{G}$, $\tilde{\mathbb{V}}$ must also contain a cutset of $\mathbb{G}$; conversely, a cutset must be an $\langle s,t\rangle$ vertex cutset for some $s$ and $t$. A detailed difference between Definition \ref{CS} and \ref{stCS} is, a cutset requires that no proper subset of it being a cutset, but an $\langle s,t\rangle$ vertex cutset does not require that.
\end{remark}

\begin{definition}\label{compress}
For a vertex subset $\mathbb{V}_q\subset \mathbb{V}$, if the induced subgraph by $\mathbb{V}_q$ is connected, following the next two steps gets the compressed graph of $\mathbb{G}$ by $\mathbb{V}_q$, denoted as $\mathbb{G}_{\mathbb{V}_q}$.\\
%1. The induced subgraph by $\mathbb{V}_q$ is connected;\\
1. Remove $\mathbb{V}_q$ from $\mathbb{G}$, add a new vertex $v$ into $\mathbb{G}-\mathbb{V}_q$;\\
%there exists a vertex $v$ in $\mathbb{G}_{\mathbb{V}_q}$ such that $\mathbb{G}_{\mathbb{V}_q}-v=\mathbb{G}-\mathbb{V}_q$;\\
2. For any vertex $\tilde{v}\in\mathbb{G}-\mathbb{V}_q$, if $\tilde{v}$ is adjacent with at least one vertex in $\mathbb{V}_q$, connect $\tilde{v}$ with $v$.\\
Here we say $\mathbb{V}_q$ compresses to node $v$.
%$\mathbb{V}_q$ forms a connected node group, and the group compresses to node $v$.
\end{definition}

For Fig. \ref{F1}, compressing $\{v_3,v_5,v_6,v_7\}$ yields Fig. \ref{F10}.

\subsection{Model formulation}

Consider an MAS consisting of $n$ agents with single-integrator dynamics:
\begin{equation}\label{model}
\dot{x}_i=u_i,~i=1,2,\cdots,n,
\end{equation}
where $x_i,u_i\in\mathbb{R}$ represent the state and the control input of agent $v_i$, respectively. Without loss of generality, the leader which can be actuated by external inputs is supposed to be $v_1$. The set of the rest agents, i.e., followers, are denoted as $\mathbb{V}_f=\{v_2,\cdots,v_n\}$. The interconnection graph of system (\ref{model}) is denoted by $\mathbb{G}$, and $\mathbb{G}-v_1$ is said to be the follower subgraph. The control inputs obey the consensus-based protocol: $u_i=\sum\limits_{j \in N_i } a_{ij}( x_j  - x_i ) + u_{o}$, where $u_{o}\in\mathbb{R}$ is the external control on the leader agent $v_1$, and $u_o=0$ when $i=2,3,\cdots,n$.

The compact form of system (\ref{model}) under the protocol is summarized as (\ref{compact}):
\begin{equation}\label{compact}
  \dot{x}=-Lx+bu,
\end{equation}
where $x=(x_1,x_2,\cdots,x_n)^T\in\mathbb{R}^n$ and $u=u_0\in\mathbb{R}$ is the external control. $L$ is the Laplacian matrix of $\mathbb{G}$ and $b=(1,0,\cdots,0)^T\in\mathbb{R}^{n}$. System (\ref{compact}) is said to be controllable if for any initial states $x(0)=x_0\in\mathbb{R}^n$ and any target states $x^*\in\mathbb{R}^n$, there exists a $u=u(t)$ and a finite time instant $T\geq0$, such that $x(T)=x^*$. In the following, we call $[b,-Lb,\cdots,(-L)^{n-1}b]$ the controllability matrix of system (\ref{compact}).

\begin{definition}\label{NF}
If MAS (\ref{compact}) is controllable, for any removal of $p$ followers ($1\leq p\leq n-1$), let $r_p$ be the minimal rank of the controllability matrix of the remaining subsystem. If $r_p=n-p$, the controllability of the original system is said to be $p$-nodes non-fragile ($p$-nodes NF). If for all $1\leq p\leq n-1$, the controllability of system (\ref{compact}) is $p$-nodes NF, we then say the controllability is strongly non-fragile (SNF). Otherwise, if it is $p$-nodes NF for all $1\leq p\leq k$, but is not $(k+1)$-nodes NF, we say the controllability is $k$-weakly non-fragile ($k$-WNF). Especially, if the controllability is not $1$-node NF, it is said to be fragile.
\end{definition}

In this paper, we assume that the leader can not be removed. Without causing misunderstanding, we say system (\ref{compact}) is SNF or $k$-WNF in the following if the controllability of system (\ref{compact}) is SNF or $k$-WNF, respectively. Clearly, SNF is another equivalent expression of $(n-1)$-WNF, and fragile can be treated as $0$-WNF. Be worth mentioning, non-fragility is discussed for controllable MASs. Besides, if an MAS is $k$-WNF, it does not mean that the controllability is not $p$-nodes NF for all $p\geq k+1$, see Example \ref{pNF}.

\begin{figure}
\begin{minipage}[t]{0.6\linewidth}
\centering
\includegraphics[width=2.5in]{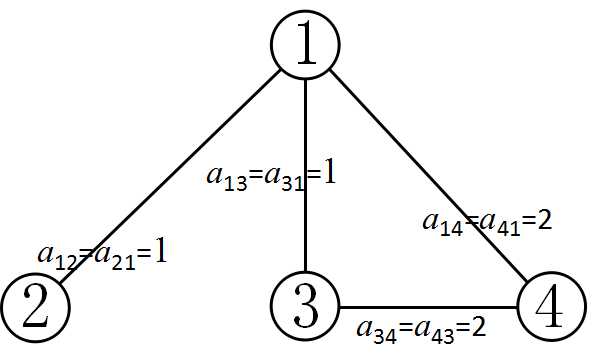}
\caption{A weighted graph with 4 nodes.}\label{F2}
\end{minipage}%
~
\begin{minipage}[t]{0.35\linewidth}
\centering
\includegraphics[width=1.6in]{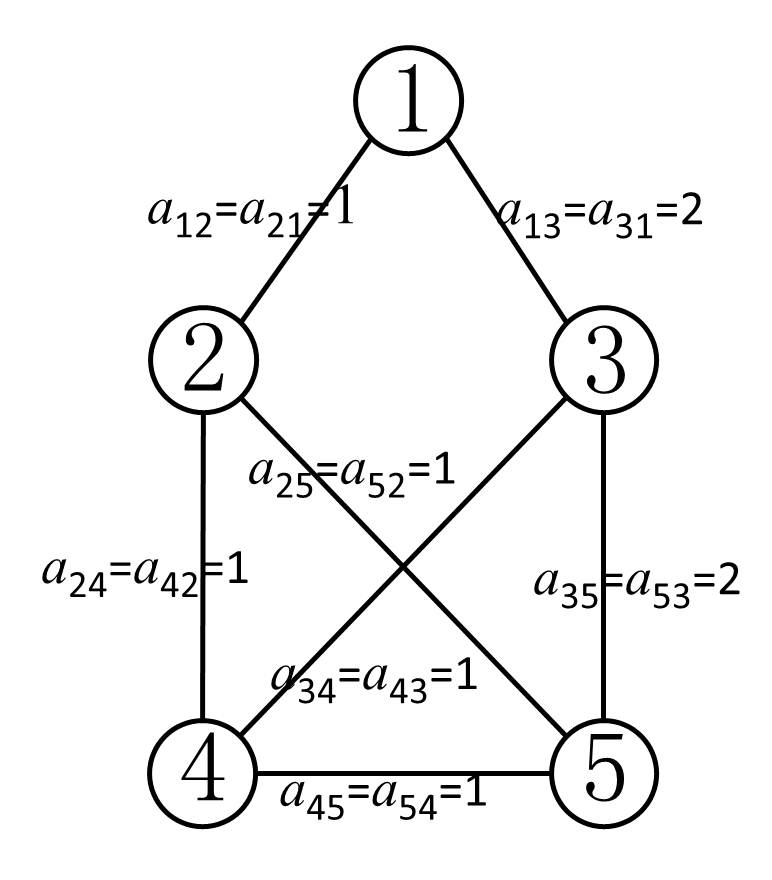}
\caption{A weighted graph with 5 nodes.}\label{F4}
\end{minipage}
\end{figure}

%\begin{figure}
%  \centering
%  % Requires \usepackage{graphicx}
%  \includegraphics[width=0.25\textwidth]{F2.jpg}\\
%  \caption{An illustrative example for Definition \ref{NF}.}\label{F2}
%\end{figure}

%\begin{figure}
%  \centering
%  % Requires \usepackage{graphicx}
%  \includegraphics[width=0.25\textwidth]{F4.jpg}\\
%  \caption{An illustrative example for Theorem \ref{spc}.}\label{F4}
%\end{figure}

\begin{example}\label{pNF}
With the interconnection topology depicted in Fig. \ref{F2}, system (\ref{compact}) is controllable. Since the system will become uncontrollable if $v_4$ is removed, the controllability is fragile. However, removing any two nodes in $\{v_2,v_3,v_4\}$ will make the remaining subsystem controllable.
\end{example}

\section{Main results on fragility of the controllability}

%Considering that losing agents changes the interaction topology, which may influence the controllability of MAS (\ref{compact}), this section discusses non-fragility of the system. Meanwhile the effect of the edge weights on the non-fragility will be analysed.

%
%As is known to all, the controllability of MAS (\ref{compact}) is only determined by the interaction topology and the couplings between the agents. Aiming at the topology change caused by losing agents, this section discusses the non-fragility of system (\ref{compact}), and analyse the influence of the edge weights on the non-fragility.

\subsection{Non-fragility of the controllability}

As defined in Definition \ref{NF}, non-fragility is a measure of controllability that how difficult it would be broken. Intuitively, the most non-fragile MAS is an SNF system, i.e., no matter how many followers are removed, the remaining subsystem is still controllable. Next we show the condition of MASs to be SNF. A basic lemma is needed.

\begin{lemma}\label{star}
For system (\ref{compact}), if the length of $\mathbb{G}$ (relative to the leader) is $1$, and the edges connected with the leader share different weights, then, there exists an $M>0$, such that when the weights of the edges in the follower subgraph are not larger than $M$, the system is controllable.
\end{lemma}
\begin{proof}
Suppose there are $s\geq0$ edges in the follower subgraph, with the weights $w_1,w_2,\cdots,w_s$. The determinant of the controllability matrix is mathematically a function of $w_1,\cdots,w_s$, denoted as $f(w_1,\cdots,w_s)$. Since the weights of the edges connected with the leader are all different, we get $\lim\limits_{w_1,\cdots,w_s\rightarrow0^+}f(w_1,\cdots,w_s)=\Delta\neq 0$ (without loss of generality, assume $\Delta>0$ in the following). Therefore, there exists an $M>0$, such that when $w_i>0,~i=1,2,\cdots,s$, $f(w_1,\cdots,w_s)>0$, i.e., the system is controllable for these choices of edge weights.
\end{proof}

\begin{theorem}\label{SNF}
For system (\ref{compact}), there exist a set of edge weights to make the system SNF if and only if the length of $\mathbb{G}$ (relative to the leader) is $1$.
\end{theorem}
\begin{proof}
Necessity: If the length of $\mathbb{G}$ is not $1$, there exists a follower node $v$ which is not a neighbor of the leader. When all the followers are removed except $v$, node $v$ can not get information from the leader, which obviously makes the system uncontrollable.\\
%This means the system is not SNF. If two followers are adjacent with the leader by the same edge weight, when all the other followers are removed, the subsystem is also uncontrollable.\\
Sufficiency: Considering that each follower is adjacent with the leader, assign different weights on $e_{12},\cdots,e_{1n}$. According to Lemma \ref{star}, each subgraph $\mathbb{G}_i$ corresponds to an $M_i>0$, such that when all the weights of the edges connecting two followers are not larger than $M_i$, the subsystem with the interconnection topology $\mathbb{G}_i$ is controllable, $i=1,2,\cdots,2^{n-1}$. Let $M=\min\{M_1,M_2,\cdots,M_{2^{n-1}}\}$, it is obvious that after removing any group of followers, the remaining subsystem is still controllable. This means system (\ref{compact}) is SNF.
\end{proof}

Theorem \ref{SNF} provides a necessary and sufficient graphic condition for the existence of weight assignments to make system (\ref{compact}) SNF. It can be further proved, there exist a set of edge weights for the follower subgraph to make the system SNF, when and only when the edges share different weights if they are connected with the leader. For the graphs whose length is more than $1$, we generalize the result to WNF controllability, and a property of minimal cutsets is shown in advance. Especially, the concept of (minimal) cutset(s) in the following should not contain the leader. For example, the minimal cutsets in Fig. \ref{F10} are $\{v_1,v_4\}$ and $\{v_2,v_3\}$. However, if system (\ref{compact}) is with the interconnection topology depicted as Fig. \ref{F10}, the minimal cutset of the topology is only $\{v_2,v_3\}$. For simplicity, the (minimal) cutset(s) of the topology is still said to be the (minimal) cutset(s) of $\mathbb{G}$.

\begin{lemma}\label{kWNF}
For system (\ref{compact}), if the minimal cutset(s) of $\mathbb{G}$ contains $k+1$ followers, then, there exist a set of edge weights such that the system is $k$-WNF.
\end{lemma}
\begin{proof}
If the minimal cutset(s) of $\mathbb{G}$ contains $k+1$ followers, for any $0\leq p\leq k$, any choice of $p$ followers removed from $\mathbb{G}$, the remaining subgraph is connected. According to \cite{Goldin13}, for the remaining subgraph, the set of edge weights those make the subsystem uncontrollable has Lebesgue measure zero in $\mathbb{R}^{|\mathbb{E}|-p}$. Therefore, for the original graph, the set of edge weights to make the subsystem uncontrollable, denoted as $W_{i}$, also has Lebesgue measure zero in $\mathbb{R}^{|\mathbb{E}|}$, $i=1,2,\cdots,C_n^p$. Since $\sum\limits_{p=1}^{k}C_n^p<2^{|\mathbb{E}|}<+\infty$, we get the set $W_u=\bigcup\limits_{p=1}^{k}\bigcup\limits_{i=1}^{C_n^p} W_i$ has Lebesgue measure zero in $\mathbb{R}^{|\mathbb{E}|}$, where $W_u$ is the set of edge weights to make the system uncontrollable by removing any choice of no more than $k$ followers. This implies that the system is $k$-WNF.
\end{proof}

\begin{corollary}
For system (\ref{compact}), the next two statements hold for almost all sets of edge weights\footnote{The edge weight sets have Lebesgue measure $1$ in $\mathbb{R}^{|\mathbb{E}|}$.}:\\
1. The system is controllable.\\
2. The system remains controllable after removing any set of followers, unless the removed followers contain a cutset of $\mathbb{G}$.
\end{corollary}
\begin{proof}
The result can be derived directly from Theorem 3 in \cite{Goldin13}.
\end{proof}

\begin{theorem}\label{WNF}
There exist a set of edge weights to make system (\ref{compact}) $k$-WNF if and only if the minimal cutset(s) of $\mathbb{G}$ contains $k+1$ followers.
\end{theorem}
\begin{proof}
Necessity: If the minimal cutset(s) of $\mathbb{G}$ contains $k+2$ (or more) followers, by Lemma \ref{kWNF}, the system is (at least) $(k+1)$-WNF, which contradicts the $k$-WNF controllability. If the minimal cutset of $\mathbb{G}$ only contains $k$ (or less) followers, by Lemma \ref{kWNF}, the system is (at most) $(k-1)$-WNF, which also contradicts the $k$-WNF controllability. This means the minimal cutset(s) of the interconnection topology contains exactly $k+1$ followers for a $k$-WNF controllable MAS.\\
Sufficiency: This can be directly derived from Lemma \ref{kWNF}.
\end{proof}

\begin{corollary}
The following assertions hold:\\
1. System (\ref{compact}) is fragile for all sets of edge weights if and only if the interconnection topology contains at least one cut vertex as a follower.\\
2. If system (\ref{compact}) is $k$-WNF, $D_1$ of the distance partition of $\mathbb{G}$ (relative to the leader) contains at least $k+1$ followers.
\end{corollary}
\begin{proof}
1. The first assertion follows directly from Definition \ref{NF} and Theorem \ref{WNF}.\\
2. According to Theorem \ref{WNF}, the minimal cutset(s) of $\mathbb{G}$ contains $k+1$ followers. If there exist only $k$ (or less) followers adjacent with the leader, when these nodes are removed, the rest followers will not be able to recieve information from the leader, which means the remaining subsystem is not controllable. This contradicts the assumption that the system is $k$-WNF.
\end{proof}

Apparently, to discuss non-fragility of controllability, system (\ref{compact}) is required to be structurally controllable, i.e., the interconnection topology is connected. However, even if the system is strongly structurally controllable (which means the system is controllable for all sets of edge weights \cite{Mayeda79}), the controllability may also be fragile. For example, system (\ref{compact}) with a path topology is strongly structurally controllable if a terminal node is selected as the leader, whereas removing any node between the terminal nodes breaks the connectedness of the graph.

%Generally speaking, the controllable star graph is the ``most non-fragile'' topology, while the line graph is the ``least'' one.

\subsection{Partial controllability}

Many MASs contain important agents, which should be controlled properly. For an uncontrollable MAS, how to ensure the important part be controllable derives the \emph{Partial Controllability Problem}. In this subsection, we provide criteria for partial controllability, and analyse the influence of losing agents on partial controllability via the concept of $\langle s,t\rangle$ vertex cutsets in graph theory. The investigation starts with the definition of controllable node groups.

\begin{definition}\label{node}
A group of agents $i_1,i_2,\cdots,i_r$ in system (\ref{compact}) are said to be partially controllable if for any initial states $x_{i_1}(t_{0}),x_{i_2}(t_{0}),\cdots,x_{i_r}(t_{0})$ and target states $x_{i_1}^*,x_{i_2}^*,\cdots,x_{i_r}^*$, there exists an external input $u(t)$ on the leader and a finite time $t_1>t_0$, such that $x_{i_1}(t_{1})=x_{i_1}^*,~x_{i_2}(t_{1})=x_{i_2}^*,~\cdots,~x_{i_r}(t_{1})=x_{i_r}^*$. Agents $i_1,i_2,\cdots,i_r$ are said to form a controllable node group if they are partially controllable. Moreover, the group is said to be maximal if adding any other agent into this group breaks its partial controllability.
%When the maximal controllable nodes group contains all nodes in (\ref{compact}), the system is said to be controllable.
\end{definition}

Apparently, system (\ref{compact}) is controllable if and only if all the nodes in $\mathbb{V}$ are partially controllable.

\begin{theorem}\label{gram}
In system (\ref{compact}), nodes in the same group $\tilde{V}=\{v_{i_1},v_{i_2},\cdots,v_{i_r}\}$ are partially controllable if and only if the principal minor formulated by the $i_1,i_2,\cdots,i_r$-th rows and columns of the Grammian matrix $W_c(t_0,t_1)$, denoted as $\tilde{W}$, is invertible, where
\begin{equation*}
  W_c(t_0,t_1)=\int_{t_0}^{t_1}e^{-L(t_1-t)}bb^Te^{-L^T(t_1-t)}dt.
\end{equation*}
\end{theorem}
\begin{proof}
Without loss of generality, suppose $i_1=1,i_2=2,\cdots,i_r=r$.\\
Sufficiency: Let $\tilde{x}(t_0)=(x_1(t_0),x_2(t_0),\cdots,x_r(t_0))^T\in\mathbb{R}^r$ and $\tilde{x}^*=(x_1^*,x_2^*,\cdots,x_r^*)^T\in\mathbb{R}^r$ be the initial states and target states of agents $v_1,v_2,\cdots,v_r$, respectively. Denote $x^*=(\tilde{x}^{*T}, 0,\cdots,0)^T\in\mathbb{R}^n$, and the initial state of the whole system is $x(t_0)\in\mathbb{R}^n$. Next we prove that there exists a $\tilde{z}\in\mathbb{R}^r$ such that when $u=-b^Te^{-L^T(t_1-t)}z$, the states of agents $v_1,v_2,\cdots,v_r$ at $t_1$ are $x_1^*,x_2^*,\cdots,x_r^*$, respectively, where $z=(\tilde{z}^T,0,\cdots,0)^T\in\mathbb{R}^n$. The trajectory of system (\ref{compact}) is
\begin{equation}\label{states}
\begin{aligned}
x({t_1})&= {e^{ - L{t_1}}}x({t_0}) + \int_{t_0}^{{t_1}} {{e^{ - L({t_1} - t )}}bu(t )dt } \\
        &= {e^{ - L{t_1}}}x({t_0}) + \int_{t_0}^{{t_1}} {{e^{-L(t_1-t) }}b( - {b^T}{e^{{-L^T}(t_1-t) }})zdt } \\
        &= {e^{ - L{t_1}}}x({t_0}) - W_c(t_0,t_1)z.\\
\end{aligned}
\end{equation}
%\end{eqnarray}
Denote $\tilde{y}$ as the first $r$ entries of $e^{-Lt_1}x(t_0)$, let $\tilde{z}=\tilde{W}^{-1}(\tilde{y}-\tilde{x}^*)$, the first $r$ entries of $W_c(t_0,t_1)z$ will be $\tilde{x}^*-\tilde{y}$. This ensures that $x_i(t_1)=x_i^*$, $i=1,2,\cdots,r$, which means agents $v_1,v_2,\cdots,v_r$ are partially controllable.\\
Necessity: If $\tilde{W}$ is not invertible, there exists a $\tilde{w}\neq0\in\mathbb{R}^r$ such that $\tilde{W}\tilde{w}=0$. Let $w=(\tilde{w}^T,0,\cdots,0)^T\in\mathbb{R}^n$, it is not difficult to check that $w^TW_c(t_0,t_1)w=0$. Since agents $v_1,v_2,\cdots,v_r$ are partially controllable, there exist $u(t)$ and $t_1>t_0$ such that
\begin{equation*}
  x({t_1}) = {e^{ - L{t_1}}}x({t_0}) + \int_{t_0}^{{t_1}} {{e^{ - L({t_1} - t )}}bu(t )dt }
\end{equation*}
holds for any $\tilde{x}(t_0)$ and any $\tilde{x}(t_1)$, regardless of the states of the other agents. As $w\neq0$, select $x(t_0)$ and $x(t_1)$ such that $w^T(x(t_1)-{e^{ - L{t_1}}}x({t_0}))\neq0$, and this yields
\begin{equation*}
  w^T\int_0^{{t_1}} {e^{ -L(t_1-t)}}bu(t )dt\neq0.
\end{equation*}
However,
\begin{equation*}
  {w^T}{W_c}w = \int_{{t_0}}^{{t_1}} {||{w^T}{e^{-L(t_1-t) }}b|{|^2}dt }
              = 0
\end{equation*}
yields
\begin{equation*}
  {w^T}{e^{-L(t_1-t) }}b=0, t_0<t\leq t_1.
\end{equation*}
Therefore
\begin{equation*}
  \int_{{t_0}}^{{t_1}} {{w^T}{e^{-L(t_1-t) }}bu(t )dt }  = 0,
\end{equation*}
which makes a contradiction, thus $\tilde{W}$ is invertible.
\end{proof}

\begin{theorem}\label{rank}
In system (\ref{compact}), nodes in the same group $\tilde{V}=\{v_{i_1},v_{i_2},\cdots,v_{i_r}\}$ are partially controllable if and only if the $i_1,i_2,\cdots,i_r$-th rows of the controllability matrix $Q=[b,-Lb,\cdots,(-L)^{n-1}b]$ are linearly independent.
\end{theorem}
\begin{proof}
Without loss of generality, suppose $i_1= 1,i_2= 2,\cdots ,i_r= r$.\\
Sufficiency: The first $r$ rows of $Q$ are linearly independent. If agents $v_1,v_2,\cdots,v_r$ are not partially controllable, according to Theorem \ref{gram}, $\tilde{W}$ is not invertible for all $t_1>t_0$, which means there exists a $\tilde{w}\neq0\in\mathbb{R}^r$ such that $\tilde{w}^T\tilde{W}\tilde{w}=0$. Let $w=(\tilde{w}^T,0,\cdots,0)^T$, we get
\begin{equation*}
  0=w^TW_c(t_0,t_1)w=\int_{t_0}^{t_1} (w^Te^{-L(t_1-t)}b)(w^Te^{-L(t_1-t)}b)^Tdt.
\end{equation*}
Therefore $w^Te^{-L(t_1-t)}b=0$  holds for all $t_0<t\leq t_1$, thus $w^TL^kb=0$, $k=0,1,2,\cdots,n-1$ (actually it holds for all $k\geq0$). Take the first $r$ rows of $Q$, denoted as $\tilde{Q}$, since $\tilde{w}\neq0$ and $\tilde{w}^T\tilde{Q}=0$, these rows must be linearly dependent, which makes a contradiction.\\
Necessity: If $rank(\tilde{Q})<r$, there exists a $\tilde{w}\neq0\in\mathbb{R}^r$ such that $\tilde{w}^T\tilde{Q}=0$. Let $w=(\tilde{w}^T,0,\cdots,0)^T\in\mathbb{R}^n$, according to Hamilton-Cayley Theorem, $w^TL^kb=0$ holds for all $k=0,1,2,\cdots$, thus $w^Te^{-L(t_1-t)}b=0,t_0<t\leq t_1$. Since $w^TW_c(t_0,t_1)w=0=\tilde{w}^T\tilde{W}\tilde{w}$, we declare that $\tilde{W}$ is not invertible, i.e., agents $v_1,v_2,\cdots,v_r$ are not partially controllable. This is a contradiction.
\end{proof}

\begin{corollary}\label{uncontrollable}
The following assertions hold:\\
1. If the rank of the controllability matrix is $r$, then there are at most $r$ agents in system (\ref{compact}) can be partially controllable.\\
2. If some rows of the controllability matrix are linearly dependent, then, the agents corresponding to these rows are not partially controllable.\\
3. Each set of maximal linearly independent rows of the controllability matrix corresponds to a maximal controllable node group.
\end{corollary}
\begin{proof}
These assertions can be obtained directly from Theorem \ref{rank}.
\end{proof}

Based on the concept of partial controllability, we should consider whether we can make the nodes in a controllable node group be still partially controllable after removing some other followers, which is called \emph{Partial Controllability Preserving Problem}. Similar to the effect of cutsets on non-fragility of controllability, the concept of $\langle s,t\rangle$ vertex cutsets plays a critical roll in preserving partial controllability.

\begin{lemma}\label{qCS}
For any weights of the edges in $\mathbb{G}$, follower $q$ is partially controllable after removing a node set $\mathbb{V}'$ if and only if $\mathbb{V}'$ is not a $\langle v_1,q\rangle$ vertex cutset.
\end{lemma}
\begin{proof}
Necessity: If $\mathbb{V}'$ is a $\langle v_1,q\rangle$ vertex cutset, according to Definition \ref{stCS}, removing $\mathbb{V}'$ makes $q$ not able to receive information from $v_1$, therefore $q$ is not partially controllable, which contradicts the assumption.\\
Sufficiency: Consider the subsystem generated by removing $\mathbb{V}'$, since $\mathbb{V}'$ is not a $\langle v_1,q \rangle$ vertex cutset, there exists a path from the leader to $q$ in $\mathbb{G}-\mathbb{V}'$. By Theorem \ref{rank}, since the row of the controllability matrix corresponding to $q$ remains not all-zero, $q$ is controllable.
\end{proof}

This lemma shows that a specific follower is always partially controllable for any weight assignments, unless the removed followers form a $\langle v_1,q\rangle$ vertex cutset. Usually, the important agent group of an MAS consists of not only a single agent. To ensure all the agents in the group be partially controllable, the concept of ``compressed graph'' is needed. Refer to Definition \ref{compress}, suppose the induced subgraph by $\mathbb{V}_q$ is connected, and $\mathbb{G}_{\mathbb{V}_q}$ is the compressed graph of $\mathbb{G}$ by $\mathbb{V}_q$.

\begin{theorem}\label{spc}
For any weights of the edges in $\mathbb{G}-{\mathbb{V}_q}$, there exist a set of weights for the other edges in $\mathbb{G}$, such that nodes in $\mathbb{V}_q$ are partially controllable after removing a node set $\mathbb{V}'$ ($\mathbb{V}_q\bigcap\mathbb{V}'=\emptyset$) if and only if $\mathbb{V}'$ is not a $\langle v_1,q\rangle$ vertex cutset of $\mathbb{G}_{\mathbb{V}_q}$, where $q$ is the node that $\mathbb{V}_q$ compresses to.
\end{theorem}
\begin{proof}
Necessity: If $\mathbb{V}'$ is a $\langle v_1,q\rangle$ vertex cutset of $\mathbb{G}_{\mathbb{V}_q}$, by Definition \ref{stCS}, none of the nodes in $\mathbb{V}_q$ receives information from the leader $v_1$, which makes the nodes in $\mathbb{V}_q$ not partially controllable.\\
Sufficiency: Consider the subgraph $\mathbb{G}_{\mathbb{V}_q}$, by Lemma \ref{qCS}, $q$ is controllable for any set of edge weights in $\mathbb{G}-\mathbb{V}_q$ (which is also $\mathbb{G}_{\mathbb{V}_q}-q$). Obviously, there exist a set of edge weights for the induced graph of $\mathbb{V}_q$ such that nodes in $\mathbb{V}_q$ are partially controllable. Similar to the sufficiency proof of Theorem \ref{SNF}, assign weights small enough to the edges between $\mathbb{G}-\mathbb{V}_q$ and $\mathbb{V}_q$, the rows of the controllability matrix corresponding to $\mathbb{V}_q$ remains linearly independent if $\mathbb{V}_q\bigcap\mathbb{V}'=\emptyset$, i.e., the partial controllability of $\mathbb{V}_q$ is preserved.
\end{proof}

%Theorem \ref{spc} indicates that there exist a set of weights for the edges between $\mathbb{V}_q$ and the other agents, to preserve the partial controllability after removing any node set, unless the removed nodes form a $\langle v_1,q\rangle$ vertex cutset of $\mathbb{G}_{\mathbb{V}_q}$.
For Theorem \ref{spc}, not all choices of edge weights could guarantee that the partial controllability be preserved (see Example \ref{fanli1}). However, it can be proved, the set of weights has Lebesgue measure zero in $\mathbb{R}^{|\mathbb{E}|}$, for the edges removing whom will make $\mathbb{V}_q$ not partially controllable. Meanwhile, in Definition \ref{compress}, the induced subgraph by $\mathbb{V}_q$ being connected should not be ignored. If it is not connected, $\mathbb{V}_q$ can not compress to a single node, otherwise, the sufficiency of Theorem \ref{spc} no longer holds, see Example \ref{fanli2}.

\begin{example}\label{fanli1}
%\begin{figure}
%  \centering
%  % Requires \usepackage{graphicx}
%  \includegraphics[width=0.25\textwidth]{F4.jpg}\\
%  \caption{An illustrative example for Theorem \ref{spc}.}\label{F4}
%\end{figure}

The interconnection topology of system (\ref{compact}) is depicted as Fig. \ref{F4}, and the system is controllable. Compressing $\{v_4,v_5\}$ yields Fig. \ref{F10} (whose edge weights are not of importance and are omitted). Obviously, $v_3$ does not form a $\langle v_1,v_4\rangle$ vertex cutset of Fig. \ref{F10}. However, if $v_3$ is removed, the remaining subsystem is not controllable.
\end{example}

\begin{example}\label{fanli2}
\begin{figure}
  \centering
  % Requires \usepackage{graphicx}
  \includegraphics[width=2.2in]{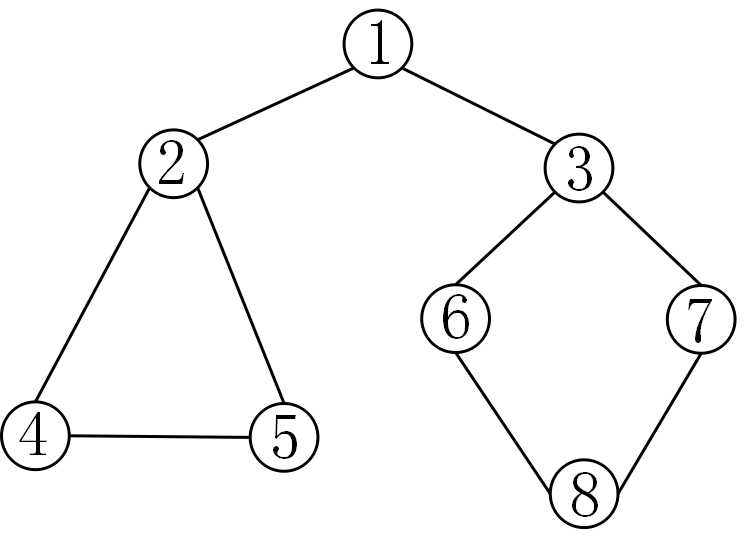}\\
  \caption{A graph consists of 8 nodes.}\label{F3}
\end{figure}

The interconnection topology of system (\ref{compact}) is depicted as Fig. \ref{F3}, and the system is structurally controllable. The induced subgraph by $\{v_4,v_5,v_6,v_7,v_8\}$ is not connected, and the compressed graph is not Fig. \ref{F10}. Otherwise, when $v_3$ is removed, the remaining subsystem should be also structurally controllable by Theorem \ref{spc}. However, removing $v_3$ makes $\{v_6,v_7,v_8\}$ separated from the leader, i.e., they are not partially controllable for any weight assignments.
\end{example}

When the induced subgraph by $\mathbb{V}_q$ is not connected, the result is as follows. Suppose that the connected components $\mathbb{V}_{q_1},\cdots,\mathbb{V}_{q_s}$, compress to $q_1,\cdots,q_s$, respectively.

\begin{theorem}\label{stCSs}
For any weights of the edges in $\mathbb{G}-{\mathbb{V}_q}$, there exist a set of weights for the other edges in $\mathbb{G}$, such that nodes in $\mathbb{V}_q$ are partially controllable after removing a node set $\mathbb{V}'$ ($\mathbb{V}_q\bigcap\mathbb{V}'=\emptyset$) if and only if $\mathbb{V}'$ is not a $\langle v_1,q_i\rangle$ vertex cutset of $\mathbb{G}_{\mathbb{V}_q}$ for any $i=1,\cdots,s$, where $q_i$ is the node that $V_{q_i}$ compresses to.
\end{theorem}
\begin{proof}
Necessity is obvious. Sufficiency proof can be derived from the proof of Theorem \ref{spc} by mathematical induction, and is omitted here.
\end{proof}

\begin{corollary}
There exist a set of edge weights for system (\ref{compact}), such that nodes in $\mathbb{V}_q$ are partially controllable after removing a node set $\mathbb{V}'$ ($\mathbb{V}_q\bigcap\mathbb{V}'=\emptyset$) if and only if $\mathbb{V}'$ is not any $\langle v_1,\tilde{v}\rangle$ vertex cutset of $\mathbb{G}$ for all $\tilde{v}\in \mathbb{V}_q$.
\end{corollary}
\begin{proof}
Consider every single node in $\mathbb{V}_q$ as the node compressed from itself, the result is followed directly from Theorem \ref{stCSs}.
\end{proof}

\begin{remark}
In summary, the problem of preserving partial controllability is equivalently transformed into the problem of finding the the minimal $\langle s,t\rangle$ vertex cutsets of the compressed interconnection graph. For any graph $\mathbb{G}$, time complexity of the search algorithm for all the minimal $\langle s,t\rangle$ vertex cutsets is $O(|\mathbb{V}|+|\mathbb{E}|)$ \cite{Patvardhan}. Therefore, partial controllability preserving problem has a polynomial-time complexity algorithm for the solution.
\end{remark}

\section{Conclusions}

Non-fragility of MASs was investigated under undirected interconnection topologies. It was proved that there exist a set of edge weights to make the system SNF if and only if each follower is directly adjacent with the leader. The necessary and sufficient condition of $k$-WNF controllability is the minimal cutset(s) of the interconnection graph contains $k+1$ followers. For partial controllability, a group of nodes are partially controllable if and only if the corresponding rows of the Grammian matrix, as well as the controllability matrix, are linearly independent. The existence of edge weights preserving partial controllability is equivalent to that the removed node set does not form any $\langle s,t\rangle$ vertex cutset, where $s,t$ represent the leader agent and the followers in the controllable node group, respectively.

% if have a single appendix:
%\appendix[Proof of the Zonklar Equations]
% or
%\appendix  % for no appendix heading
% do not use \section anymore after \appendix, only \section*
% is possibly needed

% use appendices with more than one appendix
% then use \section to start each appendix
% you must declare a \section before using any
% \subsection or using \label (\appendices by itself
% starts a section numbered zero.)
%

%\appendices
%\section{Proof of the First Zonklar Equation}
%Appendix one text goes here.

% you can choose not to have a title for an appendix
% if you want by leaving the argument blank
%\section{}
%Appendix two text goes here.

% use section* for acknowledgement

% Can use something like this to put references on a page
% by themselves when using endfloat and the captionsoff option.
\ifCLASSOPTIONcaptionsoff
  \newpage
\fi

% trigger a \newpage just before the given reference
% number - used to balance the columns on the last page
% adjust value as needed - may need to be readjusted if
% the document is modified later
%\IEEEtriggeratref{8}
% The "triggered" command can be changed if desired:
%\IEEEtriggercmd{\enlargethispage{-5in}}

% references section

% can use a bibliography generated by BibTeX as a .bbl file
% BibTeX documentation can be easily obtained at:
% http://www.ctan.org/tex-archive/biblio/bibtex/contrib/doc/
% The IEEEtran BibTeX style support page is at:
% http://www.michaelshell.org/tex/ieeetran/bibtex/
%\bibliographystyle{IEEEtran}
% argument is your BibTeX string definitions and bibliography database(s)
%\bibliography{IEEEabrv,../bib/paper}
%
% <OR> manually copy in the resultant .bbl file
% set second argument of \begin to the number of references
% (used to reserve space for the reference number labels box)

\end{document}